\documentclass[11pt,a4paper]{article}

\usepackage{amsmath,amssymb}
\usepackage{graphicx}

\newenvironment{proof}
  {\par \noindent {\bf Proof\ }}
  {\hfill $\Box$ \par \medskip}

\newtheorem{theorem}{Theorem}[section]
\newtheorem{lemma}{Lemma}[section]
\newtheorem{proposition}{Proposition}[section]
\newtheorem{remark}{Remark}[section]

\newcommand{\BE}{\begin{equation}}
\newcommand{\EE}{\end{equation}}

\newcommand{\cE}{\mathcal{E}}
\newcommand{\cF}{\mathcal{F}}

\newcommand{\cI}{\mathcal{I}}

\newcommand{\cO}{\mathcal{O}}

\newcommand{\mR}{\mathbb{R}}
\newcommand{\e}{\mathrm{e}}

\newcommand{\veps}{\varepsilon}
\newcommand{\BB}{\boldsymbol{B}}
\newcommand{\FF}{\boldsymbol{F}}

\newcommand{\ff}{\boldsymbol{f}}

\newcommand{\xx}{\boldsymbol{x}}
\newcommand{\pp}{\boldsymbol{p}}
\newcommand{\uu}{\boldsymbol{u}}

\newcommand{\ww}{\boldsymbol{w}}
\newcommand{\nnu}{\boldsymbol{\nu}}

\newcommand{\ssg}{\boldsymbol{\sigma}}
\newcommand{\aalpha}{\boldsymbol{\alpha}}
\newcommand{\mm}{\boldsymbol{m}}
\newcommand{\rrho}{\boldsymbol{\rho}}
\newcommand{\qq}{\boldsymbol{q}}
\newcommand{\JJ}{\boldsymbol{J}}
\newcommand{\zzero}{\boldsymbol{0}}
\newcommand{\bk}[1]{{\langle #1 \rangle}}
\newcommand{\abs}[1]{{\vert {#1} \vert}}

\newcommand{\pt}{\partial}
\newcommand{\eq}{\mathrm{eq}}

\DeclareMathOperator{\Li}{\mathrm{Li}}

\title{\Large \bf 
Hydrodynamic equations for electrons in graphene 
obtained from the maximum entropy principle}
\author {Luigi Barletti}
\date{
\small
Dipartimento di Matematica e Informatica ``Ulisse Dini''
\\
\small
Universit\`a degli Studi di Firenze
\\
\small
luigi.barletti@unifi.it
}
\begin{document}
\maketitle

\begin{abstract}
The maximum entropy principle is applied to the formal derivation of 
isothermal, Euler-like equations for semiclassical fermions (electrons and holes) in graphene.
After proving general mathematical properties of the equations so obtained, their asymptotic 
form corresponding to significant physical regimes is investigated. 
In particular, the diffusive regime, the Maxwell-Boltzmann regime (high temperature), 
the collimation regime and the degenerate gas limit (vanishing temperature) are considered.
\end{abstract}

\section{Introduction}
Graphene is a 2-dimensional crystal consisting of a single-layer honeycomb 
lattice of carbon atoms. 
Most of the interesting electronic properties of graphene derive from the conical shape of energy bands in the
vicinity of the so-called {\em Dirac points} in the electron pseudomomentum space.
Close to such points, electrons are described, with good approximation, by the Hamiltonian  \cite{CastroNeto09}
\BE
\label{H}
  H(\xx,\pp) = c \,\pp\cdot\ssg + V(\xx) \sigma_0.
\EE
Here, $\xx = (x_1,x_2,0)$ and $\pp = (p_1,p_2,0)$ are the coordinates of the electron position and pseudomomentum (measured relatively to the Dirac point), $c \sim 10^6\mathrm{m}/\mathrm{s}$ is the Fermi velocity, and $V(\xx)$ 
is an external/self-consistent electric potential.
Moreover, $\sigma_0$ denotes the $2\times 2$ identity matrix and, as usual, $\ssg = (\sigma_1,\sigma_2,\sigma_3)$ 
is the vector of Pauli matrices.
Note that \eqref{H} is a Dirac-like Hamiltonian \cite{Thaller92} for 2-dimensional massless particles with an 
``effective light speed'' $c$, which is about $1/300$ of the speed of light in vacuum, and subject to electric forces.
The spin-like degree of freedom associated to the Hamiltonian \eqref{H} is called ``pseudospin'' and is 
related to the decomposition of the graphene honeycomb lattice into two non-equivalent triangular lattices 
\cite{CastroNeto09} (note that, although the continuous degrees of freedom ($\xx$ and $\pp$) are 2-dimensional, 
the pseudospin vector is 3-dimensional). 
The electronic energy bands, i.e.\ the eigenvalues of $H$ evaluated at $V=0$, are given by
\BE
\label{bands}
  E_\pm(\pp) = \pm c \abs{\pp},
\EE
and have the above-mentioned conical shape. 
The semiclassical velocities associated to the energy bands \eqref{bands}  are
\BE
\label{vel}
  \nabla_{\pp} E_\pm(\pp) = \pm c\, \frac{\pp}{ \abs{\pp}},
\EE
showing that, from a semiclassical viewpoint, electrons move at constant speed $c$.
\par
\smallskip
When dealing with a statistical population of electrons in graphene, a kinetic (Boltzmann-like) approach 
can provide very accurate models  (see e.g.\ Ref.\ \cite{Lichtenberger11}) but requires a considerable
computational effort, especially in dimension 2 or 3.
However, hydrodynamic or diffusive models (where, by  ``hydrodynamic''  we mean a description 
in terms of Euler-like equations) can offer a good enough
accuracy at a much lower computational cost. 
Different kinds of hydrodynamic models for graphene are found in literature.
In Refs.\ \cite{Mueller09} and \cite{Svintsov12}, semiclassical Euler equations with Fermi-Dirac statistics 
are obtained in the linear-response approximation. 
The two papers, however, differ in the choice of the macroscopic moments that characyerize the fluid: 
in Ref.\ \cite{Svintsov12} such moments are the density and the average pseudomomentum while, 
in Ref.\ \cite{Mueller09}, the moments are the density and the average {\em direction} of pseudomomentum 
(this will be also our choice).
In Refs.\ \cite{JungelZamponi13,Zamponi12,M2AS11} quantum corrections to semiclassical fluid equations
of various types (bipolar, spinorial, diffusive, hydrodynamic) are obtained, assuming Maxwell-Boltzmann statistics;
such corrections account for quantum pressure (Bohm potential) and for quantum interference between
positive-energy and negative-energy electrons.
Finally, fully-quantum hydrodynamic equations for pure states are obtained  in Refs.\ \cite{CAIM12,Bialynicki95}; such
equations are formally equivalent to the Schr\"odinger equation and represent, therefore,
a graphene equivalent of the Madelung system.
\par
\smallskip
The purpose of the present paper is to develop a semiclassical hydrodynamic description of an isothermal gas
of electrons and holes in graphene, assuming Fermi-Dirac statistics but without the linear-response approximation.
To do so, we shall proceed ``from first principles'', the fundamental assumptions in our derivation being only the 
Hamiltonian \eqref{H}, and the maximum entropy principle (MEP), which will be used to close the system of equation
for the moments (see Sect.\ \ref{Sec_MomentClosure}).
We remark that the MEP is a valuable conceptual tool that can be invoked whenever an 
``information gap'' has to be filled; it has proven to be useful in a great variety of situations \cite{Wu97} and has been 
successfully applied to semiconductor modeling \cite{Camiola13,LaRosa09}. 
\par
The derivation of the hydrodynamic equations in their general form is presented in Section \ref{Sec2}.
The starting point is a kinetic description, which is represented by semiclassical Wigner equations \eqref{WE} 
endowed with a BGK term describing the relaxation of the system to a local-equilibrium state.
In Sect.\ \ref{Sec_MomentClosure}, by taking suitable moments of the Wigner equations, we obtain a system
of equations for the electron/hole densities and for the electron/hole direction fields.
Then, the application of MEP allows to select a local-equilibrium state which provides a closure of such system
and yields equations for electrons that are decoupled from (and, apart from the charge sign, identical to)
the equations for holes.
The model obtained in this way is summarized in Sect.\ \ref{Sec_Hyp}: it consists of isothermal 
Euler-like equations  for the density $n$ and the direction field $\uu$, Eq.\ \eqref{S1}, 
together with an implicit constitutive relation for the  higher-order moments $P_{ij}$ and $Q_{ij}$. 
Such relation is given by the fact that $P_{ij}$ and $Q_{ij}$ are computed on the MEP equilibrium state 
$f_\eq$ which depends implicitly on $n$ and $\uu$ through the constraint that  $f_\eq$ must have the moments  
$n$ and $\uu$ (Eq.\ \eqref{S4}).
In Theorem \ref{theo_entropy} we prove that system \eqref{S1} admits a strictly convex entropy 
(physically, the free-energy)  and is therefore hyperbolic. 
\par
Section \ref{Sec3} is devoted to the study of the constraint equations \eqref{S4}.
Since $f_\eq$ is parametrized by three Lagrange multipliers, $A$ and $\BB = (B_1,B_2)$, solving the constraint
equations is equivalent to inverting the map $(A,\BB) \to (n,\uu)$. 
In Sect.\ \ref{Sec_solvability}, by introducing a family of functions $\cI_N^s(A,B)$ 
(see definition \eqref{Idef}), the problem is reduced to the inversion of the map $(A,B) \to (n,\abs{\uu})$. where
$B = \abs{\BB}$.
In Theorem \ref{solvability} we prove that, indeed, such map is a global diffeomorphism.
Then, in Sect.\ \ref{secPQ}, the functions $\cI_N^s$ are used to obtain an  expression of the 
moments  $P_{ij}$ and $Q_{ij}$ as functions of $n$, $\uu$ and of the scalar Lagrange multipliers $A$ and $B$.
Finally, a series expansion of $\cI_N^s(A,B)$ is computed in Sect.\ \ref{Sec_series}, which ends the general part.
\par
\smallskip
Although an explicit expression of $P_{ij}$ and $Q_{ij}$ as functions of $n$ and $\uu$ cannot be obtained in general, 
nevertheless it can be obtained in some particular case of physical relevance, which are dealt with in Sect.\ \ref{Sec_asymp}. 
In Sect.\ \ref{Sec_diffu} we study the diffusive limit, corresponding to $\abs{\uu} \to 0$ (or, equivalently, to $B \to 0$). 
At leading order we obtain a non-standard drift-diffusion equation, Eq.\ \eqref{DEgeneral}, where the form 
of diffusion and mobility coefficients is somehow inverted with respect to the drift-diffusion equations for fermions 
with standard (parabolic) dispersion relation.   
At first-order in $B$ (which is equivalent to a linear-response approximation) we obtain the corresponding 
wave equation, Eq.\ \eqref{Wave}.
\par
In Sect.\ \ref{Sec_MB} we study the Maxwell-Boltzmann limit (i.e.\ the asymptotics for high temperature, $T \to \infty$). 
In this regime the mathematical structure simplifies significantly because the function $\cI_N^s(A,B)$ 
takes a factorized form. 
This implies that the tensors $P_{ij}$ and $Q_{ij}$ can be expressed in terms of $n$ and $\uu$ by means 
of a single scalar function $X(\abs{\uu})$ (see Eq.\ \eqref{PQMB}). 
In the Maxwell-Boltzmann regime, not only the diffusive limit, Eq.\ \eqref{DEMB}, but also the opposite
``collimation'' limit $\abs{\uu} \to 1$ can be considered, which corresponds to all  particles having  (locally) 
the same direction. 
As illustrated in Remark \ref{GeomRem}, the hydrodynamic equations for the collimation regime, Eq.\ \eqref{MBPE}, 
show the  properties  of a geometrical-optics system.
\par
Finally, in Sect.\ \ref{Sec_DG}, we obtain the asymptotic form of system \eqref{S1} for the so-called 
degenerate Fermi gas, corresponding to the limit $T\to 0$.
In this case, the tensors  $P_{ij}$ and $Q_{ij}$ can be expressed in terms of three scalar functions $Y(\abs{\uu})$,
$Z(\abs{\uu})$ and $Z_\perp(\abs{\uu})$, whose asymptotic behaviour for $\abs{\uu} \to 0$ and $\abs{\uu} \to 1$ 
is analyzed in Theorem \ref{Theo_YZ}. 
Also in this case  the diffusive limit, Eq.\ \eqref{DEDG}, as well as the collimation limit, can be considered, the latter leading
to equations where force terms completely disappear.
\section{Derivation of the hydrodynamic equations}
\label{Sec2}
\subsection{Kinetic equations}
\label{SecKE}
The starting point of our derivation is the kinetic description of a statistical electronic state
in terms of the Wigner matrix
\BE
\label{WM}
   W = w_0\sigma_0 + \ww \cdot \ssg,
\EE
here decomposed in its four, real, Pauli components $w_k = w_k(\xx,\pp,t)$, $k = 0,1,2,3$.
It follows from general considerations \cite{JSP10,ChapterKP} that the vector $\ww = (w_1,w_2,w_3)$ 
can be semiclassically interpreted as the pseudospin density.
\par
Let 
\BE
\label{nudef}
  \nnu(\pp) = \frac{\pp}{\abs{\pp}} = (\nu_1(\pp),\nu_2(\pp),0)
\EE
be the unit vector corresponding to the direction of pseudomomentum. 
Then, the two Wigner functions
\BE
   w_\pm = \frac{1}{2}\left( w_0 \pm \nnu\cdot\ww\right)
\EE 
can be (semiclassically) interpreted \cite{ChapterKP,M2AS11} as the phase-space densities of electrons with, respectively, 
positive and negative energies, i.e.\ belonging to the upper and the lower of cones \eqref{bands}.
By introducing the orthogonal decomposition of $\ww$ with respect to $\nnu$,
\BE
\label{deco}
  \ww = w_s \nnu + \ww_\perp,
\EE
we can clearly write
\BE
  w_\pm = w_0 \pm w_s
\EE
and we remark that the Wigner matrix can be equally well described either by $(w_0, \ww)$ or by
$(w_0,w_s,\ww_\perp)$ or by  $(w_+, w_-,\ww_\perp)$.
The perpendicular part $\ww_\perp$ is responsible for the quantum interference between 
the positive-energy and negative-energy states \cite{JungelZamponi13,Morandi09,MorandiSchurrer11,Zamponi12,M2AS11} 
and so, as we shall see next, it will give no contribution in the semiclassical limit.
\par
The Wigner matrix \eqref{WM} is assumed to satisfy the semiclassical Wigner equation
\cite{M2AS11}
\BE
\label{WE}
\left\{
\begin{aligned}
&\pt_t w_0 + c\nabla_{\xx}\cdot \ww + \FF\cdot\nabla_{\pp} w_0 = \frac{1}{\tau}\left(w_0^\eq - w_0\right),
\\
& \pt_t \ww + c\nabla_{\xx} w_0 + \FF\cdot\nabla_{\pp} \ww - \frac{2c}{\hbar} \pp \times \ww 
  = \frac{1}{\tau}\left(\ww^\eq - \ww\right),
\end{aligned}
\right.
\EE
where $\FF = -\nabla_{\xx} V$ denotes the external force.
Note that the non-interacting (single-particle) Hamiltonian part of the equations has been supplemented with a simple
collisional mechanism of BGK type \cite{BGK54} that makes the system relax, in a typical time $\tau$, to a 
local-equilibrium state corresponding to the Wigner matrix $W^\eq = w_0^\eq\sigma_0 +  \ww^\eq\cdot\ssg$.
\par
Indeed, when kinetic equations are used in order to derive asymptotic fluid equations (which is our purpose), 
only very general properties of the collisional operator come into play (positivity, collisional invariants, entropy dissipation)
that are guaranteed by the BGK operator, provided that the local-equilibrium state $W^\eq$ is properly  chosen \cite{Levermore96};
a physically reasonable way to make this choice, which also ensures good mathematical properties, is to resort to
the maximum entropy principle. 
The discussion of this point is postponed to next section.
\par
Equations for the Wigner functions $w_+$ and $w_-$ can be readily deduced from  \eqref{WE}
and read as follows:
\BE
\label{WE2}
\left( \pt_t  \pm c\nnu\cdot\nabla_{\xx} + \FF\cdot\nabla_{\pp}\right) w_\pm + c\nabla_{\xx}\cdot\ww_\perp
  \pm \nnu\cdot( \FF\cdot\nabla_{\pp} )\ww_\perp = \frac{1}{\tau}\left(w_\pm^\eq - w_\pm\right).
\EE
Note that $w_+^\eq$ and $w_-^\eq$ represent the local-equilibrium distributions of, respectively, positive-energy and 
negative-energy states.
However, as we shall see in the next subsection, if $W^\eq$ is an entropy maximizer, 
then $w_-^\eq$ cannot have finite moments, which is clearly due tho the unboundedness from below of the 
Hamiltonian \eqref{H}.
As usual, this fact suggest that negative-energy states should be described in terms of electron vacancies, 
i.e.\ {\em holes}. 
\par
In our framework, holes can be formally introduced by  considering a transformation of the Wigner matrix $W$ to a new 
Wigner matrix $F = f_0\sigma_0 + \ff\cdot\ssg$ such that (omitting all variables but $\pp$)
\BE
\label{transf}
  f_+(\pp) = w_+(\pp), \qquad f_-(\pp) = 1 - w_- (-\pp).
\EE
Formally, $f_-$ represents the Wigner function of holes in the lower band. 
The transformation from $W$ to $F$ is not unique (since the required property does not involve the perpendicular part)
and, for our purposes, can be simple completed by setting, e.g.,  $\ff_\perp(\pp) = \ww_\perp(\pp)$.
\par
The equations for $f_+$ and $f_-$, then, read as follows:
\BE
\label{WE3}
\left( \pt_t  + c\nnu\cdot\nabla_{\xx} \pm \FF\cdot\nabla_{\pp}\right) f_\pm + c\nabla_{\xx}\cdot\ff_\perp
  \pm \nnu\cdot( \FF\cdot\nabla_{\pp} )\ff_\perp = \frac{1}{\tau}\left(f_\pm^\eq - f_\pm\right),
\EE
where now $f_+^\eq(\pp) = w_+^\eq(\pp)$ and $f_-^\eq(\pp) = 1 - w_-^\eq(-\pp)$
are the local-equilibrium distributions of, respectively, electrons in the upper cone and holes in the lower cone.
As we shall see next, by assuming Fermi-Dirac statistics, both $f_+^\eq$ and $f_-^\eq$ have finite moments.
\subsection{Maximum entropy closure}
\label{Sec_MomentClosure}
In this paper we are concerned with a fluid description of the two populations of carriers (electrons and holes) 
based on the densities
\BE
\label{ndef}
  n_\pm = \bk{f_\pm}, 
\EE
and on the average directions of the pseudomomentum (see Eq.\ \eqref{nudef})
\BE
\label{udef}
  \uu_\pm = \frac{\bk{\nnu \,f_\pm}}{n_\pm} = (u^\pm_1,u^\pm_2,0),
\EE
so that $c\uu_\pm$ are the average velocities.
Here, $\bk{f}$ denotes the following normalized integral\footnote{%
We are working with dimensionless Wigner functions and the constant $1/(2\pi\hbar)^2$
is necessary in order to compute physical  moments \cite{ChapterKP}.
}
of any scalar of vector-valued function of $\pp \in \mR^2$:
\BE
\label{bkdef}
  \bk{f}  = \frac{1}{(2\pi\hbar)^2}\int_{\mR^2} f(\pp)\,d\pp.
\EE
It is worth remarking that the inequality
\BE
\label{uineq}
  \abs{\uu_\pm} \leq 1
\EE
holds, as it can be immediately deduced by applying Jensen inequality to Eq.\ \eqref{udef}.
\par
Before going on, we have to come back to the kinetic level and specify the form of the local-equilibrium Wigner matrix $W^\eq$.
As anticipated in the previous section, $W^\eq$ will be chosen according to the maximum entropy principle (MEP), which
stipulates that $W^\eq$ is the most probable microscopic state (i.e.\ an entropy maximizer)  compatible with the
observed macroscopic moments.
\par
In our case, the observed moments are the densities \eqref{ndef} and the average pseudomementum directions
\eqref{udef} of electrons and holes.
If, moreover, we assume that our electron system is in thermal equilibrium (e.g.\ with a phonon bath) at fixed temperature $T>0$,
then the appropriate entropy functional is the total free-energy
\BE
\label{Edef}
  \cE(W) = \int_{\mR^4} \left[ k_B T s(W) + HW \right] d\pp\,d\xx
\EE
(which has to be minimized), where $H$ is the graphene Hamiltonian \eqref{H}, $k_B$ is the Botzmann constant and 
\BE
 \label{entropy} 
  s(W) =  W\log W + (1 - W)\log(1- W)
\EE
is (minus) the Fermi-Dirac entropy function.
Therefore, according to the MEP, we shall assume that the local-equilibrium Wigner matrix $W^\eq$ minimizes $\cE$
and is subject to the constraint of assigned moments
\BE
\label{constr}
  \bk{f_\pm^\eq} = n_\pm = \bk{f_\pm},
  \qquad
  \bk{\nnu f_\pm^\eq} = n_\pm\uu_\pm = \bk{\nnu f_\pm}
\EE
(clearly, the constraints are more naturally expressed in terms of $f_+^\eq$ and $f_-^\eq$ than
in terms of $w_+^\eq$ and $w_-^\eq$).
We remark that these constraints imply that $1$ and $\nnu$ are collisional invariants for our BGK  operator
(i.e.\ $n_\pm$ and $\uu_\pm$ are conserved by collisions).
\par
The form of $W^\eq$ can be given in a partially explicitly way.
In fact, it is not difficult to show (see e.g.\ Ref.\ \cite{JSP10}) that six Lagrange multipliers, labeled as $A_\pm$ and 
$\BB_\pm = (B_1^\pm,B_2^\pm,0)$, exist such that
\BE
\label{Gform}
 w_\pm^\eq = \frac{1}{\e^{\pm\frac{c}{k_BT} \abs{\pp} - \nnu\cdot\BB_\pm  \mp A_\pm}+1},
 \qquad
 \ww_\perp^\eq = \zzero
\EE
(the choice of the signs of $A_\pm$ and $\BB_\pm$ has been made for later convenience). 
From \eqref{Gform}, \eqref{transf} and \eqref{nudef} we immediately deduce that
\BE
\label{phipm}
  f_\pm^\eq = \frac{1}{\e^{\frac{c}{k_BT} \abs{\pp} - \nnu\cdot\BB_\pm  - A_\pm}+1},
  \qquad
  \ff_\perp^\eq = \zzero.
\EE
Hence, the semiclassical local equilibrium obtained from the MEP is given by a two independent Fermi-Dirac
 distributions, one for electrons  and one for holes, parametrized by suitable Lagrange multipliers.
The Lagrange multipliers are the  necessary degrees of freedom allowing the constraints \eqref{constr} to be fulfilled.
If the equations \eqref{constr} are solved in for $A_\pm$ and $\BB_\pm$ in terms of $n_\pm$ and $\uu_\pm$, then the 
equilibrium distributions $f_\eq^\pm$ can be taught as being parametrized by the moments $n_\pm$ and $\uu_\pm$ 
(and the temperature $T$).
This issue will will be considered in details in Sections \ref{Sec3} and \ref{Sec_asymp}.
\par
\smallskip
Let now assume that the time-scale over which the system is observed is very large compared to $\tau$.
Then, we can assume that the system is found in the local-equilibrium state
(this fact could be more rigorously justified by means of the Hilbert expansion method \cite{Cercignani88}).
Then, we rewrite Eq.\ \eqref{WE3} with $f_\pm = f_\pm^\eq$ and $\ff_\perp = \ff_\perp^\eq = \zzero$,
which yields
\BE
\label{LVL}
\left( \pt_t  + c\nnu\cdot\nabla_{\xx} \pm \FF\cdot\nabla_{\pp}\right) f_\pm^\eq  = 0.
\EE
Closed equations for $n_\pm$ and $\uu_\pm$ are simply obtained by taking the moments of Eq.\ \eqref{LVL}
and using the constraints \eqref{constr}:
\BE
\label{ME}
\left\{
\begin{aligned}
&\pt_t n_\pm  + c\pt_i (n_\pm u^\pm_i) = 0,
\\[3pt]
&\pt_t (n_\pm u^\pm_i) + c\pt_j P_{ij}^\pm  = \pm F_j Q_{ij}^\pm,
\end{aligned}
\right.
\EE
where we shortened the notations by putting $\pt_i \equiv \pt/\pt x_i$ and adopting the convention of summation over 
repeated indices.
The tensors $P_{ij}^\pm$ and $Q_{ij}^\pm$ are given by
\BE
\label{PTdef}
  P_{ij}^\pm =  \bk{\nu_i\nu_j f_\pm^\eq}
  \quad \text{and} \quad
  Q_{ij}^\pm =  \bk{ \frac{\pt \nu_i}{\pt p_j} f_\pm^\eq} = \bk{ \frac{1}{\abs{\pp}}\nu_i^\perp\nu_j^\perp f_\pm^\eq},
\EE
where\footnote{%
Although we have used the same notation for $\nnu_\perp$ and $\ww_\perp$, the former denotes a rotated unit vector, 
the latter an orthogonal projection.}
\BE
  \nnu_\perp = (-\nu_2,\nu_1,0).
\EE
The balance equations \eqref{ME} are formally closed if the distributions $f_\pm^\eq$ are thought 
as being parametrized by $n_\pm$ and $\uu_\pm$, as discussed above.
We remark, moreover,  that the equations for electrons and holes are completely decoupled. 
In fact, both in system \eqref{ME} and in the constraints  \eqref{constr}, 
the equations for the $+$ quantities do not depend on the $-$ quantities, and vice versa.
This was expected because, as already remarked, the quantum interference terms, that are the source of coupling 
in the quantum fluid descriptions models \cite{JungelZamponi13,Zamponi12,M2AS11}, disappear in the present 
semiclassical picture.
\begin{remark}
\rm
A coupling mechanism between electrons and holes can be very naturally 
introduced by considering the potential $V$ to have an
internal (mean-field) part $V_\mathrm{int}$ subject to the Poisson-like equation
\BE
\label{Poisson}
  \gamma\,(-\Delta)^{1/2} V_\mathrm{int} = n_+ - n_-,
\EE 
where $\gamma$ is a positive physical constant and the fractional Laplacian is suited to describe 
charges concentrated in a plane \cite{ElHajjMehats13}.
The signs in Eq.\ \eqref{Poisson} 
are determined by the fact that $V$ is the electron energy and $n_\pm$ are numerical densities.
\end{remark}
\subsection{Free-energy balance and hyperbolicity}
\label{Sec_Hyp}
Since  equations \eqref{constr} and \eqref{ME} are formally identical for electrons and holes 
(with the only exception of a sign in the force term), then in the subsequent discussions we need not to 
distinguish  the two populations any more and, therefore, we shall drop the $\pm$ labels everywhere.
Moreover, since 3-dimensional quantities were only important at the kinetic level, that we have now definitively
abandoned, we can henceforth consider any independent or dependent vector variable as 2-dimensional, e.g.
$$
  \pp = (p_1,p_2), \quad \xx = (x_1,x_2), \quad \nnu = (\nu_1,\nu_2), \quad \uu = (u_1,u_2)
  \quad \BB = (B_1,B_2).
$$ 
We therefore summarize the picture emerged so far as follows:
\BE
\label{S1}
\left\{
\begin{aligned}
&\pt_t n  + c\pt_i (n u_i) = 0,
\\[3pt]
&\pt_t (n u_i) + c\pt_j P_{ij}  = \pm F_jQ_{ij},
\end{aligned}
\right.
\EE
where
\BE
\label{S2}
  P_{ij} =  \bk{\nu_i\nu_j f_\eq},
  \qquad 
  Q_{ij} =  \bk{ \frac{1}{\abs{\pp}}\nu_i^\perp\nu_j^\perp f_\eq},
\EE
and where
\BE
\label{S3}
  f_\eq = \frac{1}{\e^{\frac{c}{k_BT} \abs{\pp} - \nnu\cdot\BB  - A}+1}
\EE
is subject to the constraints
\BE
\label{S4}
  \bk{f_\eq} = n,
  \qquad
  \bk{\nnu f_\eq} = n\uu.
\EE
\par
System \eqref{S1}--\eqref{S4} has some general features which are shared by other models obtained
from entropy minimization \cite{Levermore96}, the most significant being the existence of a local entropy
and the consequent hyperbolic character of the system.
\begin{theorem}
\label{theo_entropy}
Let  $s$ be the entropy function given by \eqref{entropy} and 
\BE
  \veps = k_BT s(f_\eq) + (c\abs{\pp} \pm V)f_\eq
\EE
be the microscopic free-energy associated to the local-equilibrium state $f_\eq$
(as usual, $+$ refers to electrons and $-$ to holes).
Then, the local free-energy $\bk{\veps}$ is a strictly convex entropy for system \eqref{S1},
which is therefore hyperbolic.
Moreover, $\bk{\veps}$ satisfies the balance law
\BE
\label{entropylaw}
  \pt_t\bk{\veps} + c\pt_i \bk{\nu_i \veps} = \pm n\, \pt_tV. 
\EE
\end{theorem}
\begin{proof}
For the sake of simplicity, we put $c = 1$ and $k_BT = 1$ throughout this proof.
Moreover, let us introduce more concise notations by defining the vectors
$$
  \mm = \begin{pmatrix} 1  \\  \nu_1 \\ \nu_2  \end{pmatrix},
  \quad
  \rrho = \bk{ \mm f_\eq} =  \begin{pmatrix} n \\  nu_1 \\ nu_2  \end{pmatrix},
  \qquad
  \aalpha = \begin{pmatrix} A \pm V \\  B_1 \\ B_2  \end{pmatrix},
 $$
 and the functions
 $$
  E = \abs{\pp} \pm V,
  \qquad
  M = \aalpha^T \mm,
$$
so that \eqref{S3} can be rewritten as 
\BE
\label{Faux}
  f_\eq =  \frac{1}{\e^{E - M} +1} = (s')^{-1} \left(-E + M\right),
\EE
and the moment system \eqref{S1} can be written in the form
\BE
\label{Saux}
   \pt_t \rrho + \pt_i \JJ_i = \qq,
\EE
where
$$
  \JJ_i = \bk{\nu_i \mm f_\eq}.
  \qquad
  \qq = \pm F_j \bk{\frac{\pt \mm}{\pt p_j} f_\eq}.
$$
If $\xi$ is a generic variable of $f_\eq$, from \eqref{Faux} we obtain the identity
\BE
\label{aux}
  \frac{\pt \veps}{\pt \xi} = \left( s'(f_\eq) + E \right)\frac{\pt f_\eq}{\pt \xi} + \frac{\pt E}{\pt \xi} f_\eq 
  = M \frac{\pt f_\eq}{\pt \xi} + \frac{\pt E}{\pt \xi} f_\eq,
\EE
which, for $\xi = \rrho$, gives
\BE
\label{LT1}
  \frac{\pt \bk{\veps}}{\pt \rrho}  = \aalpha^T \frac{\pt}{\pt \rrho}\bk{\mm f_\eq} = \aalpha.
\EE
By using \eqref{aux} and \eqref{LT1}, we also obtain
\BE
\label{HY1}
  \frac{\pt}{\pt \rrho} \bk{\nu_i \veps} = \bk{\nu_i \aalpha^T\mm \frac{\pt}{\pt \rrho} f_\eq}
  = \aalpha^T \frac{\pt}{\pt \rrho}\bk{\nu_i\mm f_\eq}
  = \left(\frac{\pt \bk{\veps}}{\pt \rrho}\right)^T \frac{\pt \JJ_i}{\pt \rrho} .
\EE
Relation \eqref{HY1} implies that \eqref{Saux} is an entropy for the system,
provided that $\bk{\veps}$ is a convex function of $\rrho$ (see e.g.\ Ref.\ \cite{Chen05}).
The simplest way to prove the convexity of $\bk{\veps}$ is observing that it is
the Legendre transform of the convex function 
$$
  \bk{\veps}^*(\aalpha) 
  = \aalpha^T\rrho- \bk{\veps}(\rrho)
  = \bk{\aalpha^T\mm f_\eq} - \bk{\veps}(\rrho)
  = -\bk{s(f_\eq) + (E-M)f_\eq}
$$
(note that $\aalpha$ and $\rrho$ are Legendre-conjugate variables, by Eq.\ \eqref{LT1}).
The convexity of $\bk{\veps}^*$ is evident upon writing
$$
  \frac{\pt\bk{\veps}^*}{\pt\aalpha}  = \bk{\mm f_\eq} = \rrho,
$$ 
so that the Hessian matrix of $\bk{\veps}^*$ is given by
\BE
\label{HY2}
  \frac{\pt}{\pt \aalpha}  \left(\frac{\pt  \bk{\veps}^*}{\pt \aalpha}\right)^T 
  = \bk{\mm  \left(\frac{\pt f_\eq}{\pt\aalpha}\right)^T}
  =  \bk{\mm  \mm^T \frac{\e^{E - M}}{(\e^{E - M}+1)^2} },
\EE
which is, clearly, a positive-definite matrix.
Then, $\bk{\veps}$ is a strictly convex entropy for system \eqref{Saux}, implying that the system is symmetrizable 
and, therefore, hyperbolic.
\par
In order to prove \eqref{entropylaw}, we resort again to Eq.\ \eqref{aux}  which, for $\xi = t$, yields
$$
  \pt_t \veps = M\pt_t f_\eq \pm \pt_t V f_\eq = M \left\{E, f_\eq\right\}  \pm \pt_t V  f_\eq,
$$
where we used the fact that $f_\eq$ satisfies Eq.\ \eqref{LVL}, which is a Liouville equation with Hamiltonian $E$
(and $\left\{\cdot, \cdot\right\}$ is the Poisson bracket).
Using again \eqref{aux}, it is not difficult to prove that
$$
   M \left\{E, f_\eq\right\} = \left\{E, \veps\right\}, 
$$
which shows that $\veps$ satisfies
$$
  \pt_t \veps +c\nnu\cdot \nabla_{\xx} \veps \mp \nabla_{\xx} V \cdot \nabla_{\pp} \veps = \pm \pt_t V  f_\eq.
$$
Integrating the last equation over $\pp \in \mR^2$ yields Eq.\ \eqref{entropylaw}.
\end{proof}
We remark that Eq.\ \eqref{entropylaw} implies that in an isolated region $\Omega \subset \mR^2$ 
(no free-energy flux through $\pt\Omega$) the total free-energy balance  is
\BE
   \frac{d}{dt} \int_\Omega \bk{\veps_++\veps_-} \,d\xx = \int_\Omega (n_+-n_-)\pt_tV \,d\xx.
\EE
\section{Study of the constraint equations}
\label{Sec3}
\subsection{Theorem of solvability}
\label{Sec_solvability}
Our goal is now writing in a more explicit way the moment equations \eqref{S1}, that is expressing
the moments $P_{ij}$ and $Q_{ij}$ as functions of $n$ and $\uu$.
In fact, $P_{ij}$ and $Q_{ij}$ are defined by \eqref{S2} and \eqref{S3} in terms of the Lagrange multipliers 
$A$ and $\BB$, which are related to $n$ and $\uu$ by the constraints \eqref{S4}.
\par
The first thing we have to do is computing the expressions of the moments $\bk{f_\eq}$ and $\bk{\nnu f_\eq}$ 
as functions of $A$ and $\BB$.  
By using the polar coordinates
$p_1 = \abs{\pp}\cos\theta$, $p_2 = \abs{\pp}\cos\theta$, and  defining
\BE
\label{nT}
  n_T = \frac{(k_BT)^2}{2\pi \hbar^2 c^2},
\EE
we obtain
\BE
\label{inte}
\begin{aligned}
 &\bk{f_\eq} = \frac{n_T}{\pi} \int_0^{\pi} \phi_2(A+B\cos\theta)\,d\theta 
 \\[2pt]
 &\bk{\nu_1f_\eq} = \bk{\cos\theta f_\eq} = \frac{n_T\cos\theta_B}{\pi}  \int_0^{\pi} 
 \cos\theta\,\phi_2(A+B\cos\theta)\,d\theta, 
 \\[2pt]
 &\bk{\nu_2f_\eq} = \bk{\sin\theta f_\eq} = \frac{n_T\sin\theta_B}{\pi}  \int_0^{\pi} 
 \cos\theta\,\phi_2(A+B\cos\theta)\,d\theta,
\end{aligned}
\EE
where 
\BE
\label{Bdef}
\BB = (B\cos\theta_B,B\sin\theta_B)
\EE
and 
\BE
\label{phidef}
  \phi_s(x) = \frac{1}{\Gamma(s)} \int_0^\infty \frac{t^{s-1}}{e^{t-x} + 1}\,dt
\EE
is the so-called {\em Fermi integral} of order $s>0$.
It is therefore natural to introduce, for $s>0$ and $N$ integer, the following functions:
\BE
\label{Idef}
 \cI_N^s(A,B) = \frac{1}{\pi}  \int_0^{\pi}  \cos(N\theta)\,\phi_s(A+B\cos\theta)\,d\theta,
 \qquad
 A \in \mR,\ B\geq 0.
\EE
\par
Using \eqref{inte} and \eqref{Idef}, the constraint equations \eqref{S4} can be rewritten as
\BE
\label{constr2}
 n_T \cI_0^2 = n, 
 \qquad
 n_T \cos\theta_B \cI_1^2 = nu_1,
 \qquad
 n_T \sin\theta_B \cI_1^2 = nu_2,
\EE
from which it immediately follows that
\BE
\label{thetaB}
  \cos\theta_B = \frac{u_1}{\abs{\uu}}, \qquad  \sin\theta_B = \frac{u_2}{\abs{\uu}}
\EE
(i.e.\ the direction of $\BB$ coincides with the direction of $\uu$), and that the scalar
functions $A$ and $B$ are related to $n$ and $\abs{\uu}$ by
\BE
\label{ABeq}
  \cI_0^2(A,B) = \frac{n}{n_T}, \qquad  \frac{\cI_1^2(A,B)}{ \cI_0^2(A,B)} = \abs{\uu}.
\EE
\begin{lemma}
\label{Asym}
The functions $\cI_N^s$ have the following asymptotic behavior as  $A^2+B^2 \to \infty$:
\BE
\label{BFasym}
\cI_N^s(A,B) \sim
\left\{
\begin{aligned}
   & \e^A\,I_N(B),&  &\text{if $A < -B$},
\\[4pt]
  &\frac{1}{\pi\Gamma(s+1)} \int_0^{C(A,B)} \hspace{-16pt} \cos(N\theta) (A+B\cos\theta)^s d\theta,
  & &\text{if $A > -B$},
\end{aligned}
\right.
\EE
where $I_N$ are the modified Bessel functions of the first kind and 
\BE
\label{Cdef}
  C(A,B)
   = \left\{
  \begin{aligned}
  &\arccos\left(-\textstyle{\frac{A}{B}}\right),& &\text{if $-B < A < B$},
  \\
  &\pi,& &\text{if $A \geq B$},
  \end{aligned}
  \right.
  \EE
(and $f \sim g$ means $f/g \to 1$). 
\end{lemma}
\begin{proof}
It is well known \cite{Wood92} that
\BE
\label{phias}
\phi_s(x) \sim
\left\{
\begin{aligned}
   &\e^{x},&  &\text{as $x\to -\infty$},
\\
   &\frac{x^s}{\Gamma(s+1)},&  &\text{as $x\to +\infty$},
\end{aligned}
\right.
\EE
for every $s>0$.
When $A<-B$, the argument $A+B\cos\theta$ in \eqref{Idef} is negative for all $\theta \in [0,\pi]$ and, 
then, form \eqref{phias} we have that
$$
  \cI_N^s(A,B) \sim \frac{\e^A}{\pi} \int_0^\pi \cos(N\theta)\, \e^{B\cos\theta} d\theta = \e^A\,I_N(B)
$$
as $A^2+B^2 \to \infty$.
When $A^2+B^2 \to \infty$ with $A > -B$, according again to \eqref{phias}, 
$\phi_s(A+B\cos\theta)$ will be infinitesimal for $A+B\cos\theta < 0$ and asymptotic to 
$\frac{1}{\Gamma(s+1)}(A+B\cos\theta)^s$ for $A+B\cos\theta > 0$.
Then,
$$
   \cI_N^s(A,B) \sim \frac{1}{\pi\Gamma(s+1)} \int_{[0,\pi]\cap\{A + B\cos\theta > 0\}} \cos(N\theta) (A+B\cos\theta)^s  d\theta, 
$$
which yields the second case of Eq.\ \eqref{BFasym}.
\end{proof}
We now prove that system \eqref{ABeq} has a unique solution for $n>0$ and $\abs{\uu} \in [0,1)$.
\begin{theorem}
\label{solvability}
The map
\BE
\label{map}
(A,B)\in\mR\times[0,+\infty) \longmapsto \left(\cI_0^2, \frac{\cI_1^2}{\cI_0^2}\right) \in (0,+\infty)\times[0,1)
\EE
is a global diffeomorphism.
\end{theorem}
\begin{proof}
Putting $n_T = 1$ and $\abs{\uu} = u$, we can adopt a handy notation and denote the image of the map 
by $(n,u)$, i.e.
$$
  n = \cI_0^2(A,B), \qquad u = \frac{\cI_1^2(A,B)}{\cI_0^2(A,B)}\,.
$$
A representation of the map $(A,B) \mapsto (n,u)$ is given in Figure \ref{figure1}.
We first note that the smoothness of the map follows from the fact that $\phi_s$ are analytic functions 
(for $s>0$).
Moreover, the properties
$$
  \cI_0^2(A,B) > 0, \qquad  0 \leq \cI_1^2(A,B) < \cI_0^2(A,B)
$$
can be easily verified from definition \eqref{Idef} and show that $(n,u)$ vary in $(0,+\infty)\times[0,1)$ 
as $(A,B)$ vary in $\mR\times[0,+\infty)$
(in particular, note that $\cI_1^2 = 0$, and then $u =0$, for $B=0$).
\par
Let us now prove local invertibility, which is clearly equivalent to the local invertibility of the map
$(A,B) \mapsto (\cI_0^2,\cI_1^2) =  (n,nu)$.
By using the fundamental identity
\BE
\label{deriphi}
  \frac{d}{dx}\phi_s(x) = \phi_{s-1}(x),
\EE
it is easily seen that the Jacobian matrix of this map is
$$
\begin{aligned}
&\pt_A \cI_0^2 = \cI_0^1 = \mathrm{E}[1],&\quad  &\pt_B \cI_0^2 = \cI_1^1 = \mathrm{E}[\cos\theta],
\\[4pt]
&\pt_A \cI_1^2 = \cI_1^1 = \mathrm{E}[\cos\theta],&\quad
&\pt_B \cI_1^2 = \mathrm{E}[\cos^2\theta],
\end{aligned}
$$
where we introduced the notation
$$
  \mathrm{E}[f(\theta)] = \frac{1}{\pi}  \int_0^{\pi}  f(\theta)\,\phi_1(A+B\cos\theta)\,d\theta
$$
(for fixed $A$ and $B$).
The Jacobian determinant, therefore,  is  
$$
  \mathrm{E}[1]\,\mathrm{E}[\cos^2\theta] - \mathrm{E}^2[\cos\theta] > 0,
$$
because it  is proportional to the variance of $\cos\theta$ for the probability distribution
$\phi_1(A+B\cos\theta)/\left(\pi \mathrm{E}[1]\right)$.
\par
According to Hadamard theorem, a locally invertible map (with simply connected image) 
is global invertible if and only if it is proper, which, in the present case, means the following:
for every sequence $(A_k,B_k) \in \mR\times[0,+\infty)$ such that $A_k^2+B_k^2 \to \infty$, 
every set of the form $[n_m,n_M] \times[0,u_M] \subset (0,+\infty)\times[0,1)$ contains only
a finite number of points of the image $(n_k,u_k)$ of $(A_k,B_k)$.
Now, since $A_k^2+B_k^2 \to \infty$, we can approximate our map with its asymptotic representation 
given by Lemma \ref{Asym}.
\par
In the region $A < -B$, according to \eqref{BFasym},  for large $A^2+B^2$ we can write
$$
  \cI_0^2(A,B) \sim \e^A\,I_0(B), \qquad \cI_1^2(A,B) \sim \e^A\,I_1(B).
$$
Since $I_1(B)/I_0(B)$ is an increasing function of $B$ which tends to $1$ as $B\to+\infty$, the level lines
of $u$ in this region are (asymptotically) parallel to the $A$ axis (see Figure \ref{figure1}), with values of $u$ increasing
from 0 to 1 as $B$ increases.
\begin{figure}[t]
\includegraphics[width=\linewidth]{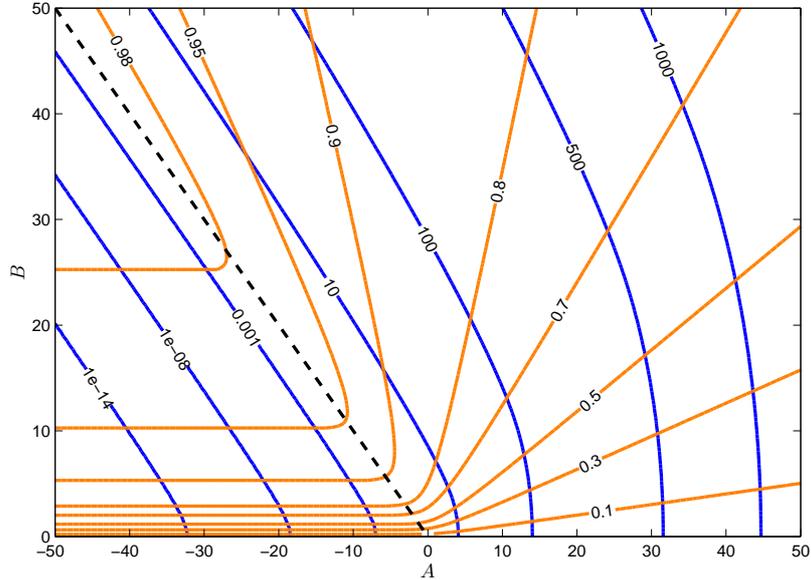}
\caption{A representation of the map $(A,B) \mapsto (n,u)$, where $n = \cI_0^2(A,B)$ and 
$u =  \cI_1^2(A,B)/\cI_0^2(A,B)$. 
The level lines of $n$ are in blue (color online) and are labeled from $10^{-14}$ to $1000$; the level
lines of $u$ are in orange (color online) and are labeled from $0.1$ to $0.98$. 
The dashed black line is the critical line  $A= -B$.}
\label{figure1}
\end{figure}
Moreover, since 
$$
  \frac{\pt_B \cI_0^2(A,B)}{\pt_A \cI_0^2(A,B)} \sim \frac{I_1(B)}{I_0(B)} = u,
$$
the level lines of $n$, as $B$ increases, tend to be parallel to the critical line $A = -B$
for increasing values of $n$.
\par
In the region $A > -B$, by \eqref{BFasym} we can write 
$$
    \cI_0^2 \sim \frac{1}{2\pi} \int_0^{C} (A+B\cos\theta)^2 d\theta, 
    \quad 
    \cI_1^2 \sim \frac{1}{2\pi} \int_0^{C} \cos\theta\,(A+B\cos\theta)^2 d\theta,
$$
with $C = C(A,B)$ given by \eqref{Cdef}.
Since these  are homogeneous functions of $(A,B)$ of degree $2$, then  $\cI_1^2/ \cI_0^2$ 
depend (asymptotically) only on the direction of the vector $(A,B)$ and, therefore, the level lines of $u$ are 
straight lines from the origin.
It can be directly verified that the values of $u$ increase counterclockwise from 0 (corresponding to $B=0$) to 1
(corresponding to $B=-A$, see Figure \ref{figure1}).
Moreover, from the above expression we get
$$
   \pt_A \cI_0^2 \sim AC + B\sin C, \qquad  \pt_B \cI_0^2 \sim \frac{1}{2}\left( BC + A\sin C\right),
$$
by which it is easy to show that the angle between $(\pt_An,\pt_Bn)$ and $(\pt_Au,\pt_Bu)$ is strictly positive
and monotonically decreasing from $\pi/2$ to 0 as $B$ goes from from $0$ to $-A$.
\par
Let us now consider a sequence $(A_k,B_k) \in \mR\times[0,+\infty)$ such that $A_k^2+B_k^2 \to \infty$.
If infinitely many points of the image sequence $(n_k,u_k)$ are such that $n_k \in [n_m,n_M]$,
then the corresponding infinitely many points of the sequence $(A_k,B_k)$ lie between the two 
level lines $n=n_m$ and $n=n_M$. 
But then, from the above discussion, we can deduce%
\footnote{It is necessary to distinguish three cases:
when both the level lines $n=n_m$ and $n=n_M$ are in the region $A<-B$, when both are in the region $A>-B$,
and when $n=n_m$ is in the first region while $n=n_M$ is in the second one.
Note, in fact, that the level lines of $n$ cannot cross (asymptotically) the critical line $A=-B$.}
that such points are forced to cross level lines of $u$ with strictly increasing values of $u$.
Hence, only finitely many points of $(n_k,u_k)$ will be contained in $[n_m,n_M] \times[0,u_M]$, which proves 
that the map is proper.
 \end{proof}
\subsection{Expression of the moments $P_{ij}$ and $Q_{ij}$}
\label{secPQ}
We shall now find an expression of the tensors $P_{ij}$ and $Q_{ij}$ (see definitions \eqref{S2} and
\eqref{S3}) in terms of the moments $n$ and $\uu$, and of the 
functions  $\cI_N^s = \cI_N^s(A,B)$.
Recalling also the definitions \eqref{bkdef}, \eqref{nT} and \eqref{phidef}, we have
$$
\begin{gathered}
  P_{ij} = \frac{1}{(2\pi\hbar)^2}\int_0^{2\pi} \int_0^{+\infty}  
  \frac{\nu_i\nu_j \abs{\pp}\,d\abs{\pp}\,d\theta}{\e^{\frac{c}{k_BT} \abs{\pp} - B\cos(\theta-\theta_B)  - A}+1}
\\[4pt]
  = \frac{n_T}{2\pi} \int_0^{2\pi} \nu_i(\theta+\theta_B)\, \nu_j(\theta+\theta_B)\,\phi_2(A+B\cos\theta) \,d\theta,
\end{gathered}
$$
where $\nu_1(\theta) = \cos\theta$ and $\nu_2(\theta) = \sin\theta$.
Using \eqref{thetaB}, we obtain
$$
  \nu_i(\theta+\theta_B)\, \nu_j(\theta+\theta_B) = 
  \frac{u_iu_j}{\abs{\uu}^2}\cos^2\theta + \frac{u_i^\perp u_j^\perp}{\abs{\uu}^2}\sin^2\theta,
$$
where
\BE
\label{uperp}
\uu^\perp = (-u_2,u_1).
\EE
Then, using $\cos^2\theta = \frac{1+\cos(2\theta)}{2}$ and $\sin^2\theta = \frac{1-\cos(2\theta)}{2}$, we can write
$$
\begin{gathered}
  P_{ij} = \frac{n_T}{\pi \abs{\uu}^2} \int_0^\pi \left[u_iu_j \frac{1+\cos(2\theta)}{2} 
  +  u_i^\perp u_j^\perp\frac{1-\cos(2\theta)}{2}\right] \phi_2(A+B\cos\theta) \,d\theta
\\[4pt]
  = \frac{n}{\abs{\uu}^2}\left( \frac{\cI_0^2+\cI_2^2}{2\cI_0^2}\, u_i u_j 
              + \frac{\cI_0^2-\cI_2^2}{2\cI_0^2}\, u_i^\perp u_j^\perp   \right),
\end{gathered}
$$
where we also used the first of equations \eqref{ABeq}.
\par
An analogous expression for $Q_{ij}$ can be readily written since, according to definition \eqref{S2},
the only changes with respect to the above expression for $P_{ij}$ are that $\nu_i$ is substituted by $\nu_i^\perp$
(and, correspondingly, $u_i$ by $u_i^\perp$),
and the extra $\abs{\pp}$ at the denominator makes the degree of $\phi_s$ decrease from $s=2$ to $s=1$.
In conclusion, we can state the following.
\begin{proposition}
The tensors $P_{ij}$ and $Q_{ij}$ have the following expressions:
\BE
\label{PQ}
\begin{aligned}
  &P_{ij} = \frac{n}{\abs{\uu}^2}\big( P\,  u_i u_j  + P_\perp u_i^\perp u_j^\perp \big),
  \\[6pt]
  &Q_{ij} = \frac{c}{k_BT}\frac{n}{\abs{\uu}^2}\big(Q\, u_i u_j + Q_\perp u_i^\perp u_j^\perp \big),
\end{aligned}
\EE
where the scalar functions $P(A,B)$, $P_\perp(A,B)$, $Q(A,B)$ and $Q_\perp(A,B)$ are given by
\BE
\label{PQcoeff}
\begin{aligned}
  &P = \frac{\cI_0^2+\cI_2^2}{2\cI_0^2},
  &\quad
  &P_\perp = \frac{\cI_0^2-\cI_2^2}{2\cI_0^2} = 1 - P,
  \\[6pt]
  &Q = \frac{\cI_0^1-\cI_2^1}{2\cI_0^2},
  &\quad
  &Q_\perp = \frac{\cI_0^1+\cI_2^1}{2\cI_0^2}.
  \end{aligned}
\EE
Moreover, the following inequalities hold:
\BE
\label{PQineq}
  \frac{1}{2} \leq P < 1,
  \qquad
  0 < Q < \frac{1}{2},
  \qquad
  0 < Q_\perp < 1-Q.
 \EE
\end{proposition}
\begin{proof}
It only remains to prove the inequalities \eqref{PQineq}.
Since $\phi_2(z)$ is strictly increasing with $z$ \cite{Lewin81}, it easily follows that $\cI_2^2 \geq 0$ and, then, 
$P \geq 1/2$.
This inequality is not strict since, from
\BE
\label{Bat0}
   \cI_N^s(A,0) = \phi_s(A) \delta_{N0},
\EE
we have that $\cI_2^2(A,0) = 0$ and then $P = 1/2$ for $B=0$.
Moreover, 
$$
  \frac{\cI_0^2+\cI_2^2}{2} =  \frac{1}{\pi}  \int_0^{\pi}  \cos^2(\theta)\,\phi_2(A+B\cos\theta)\,d\theta
  < \frac{1}{\pi}  \int_0^{\pi}  \phi_2(A+B\cos\theta)\,d\theta = \cI_0^2,
$$
which proves $P < 1$.
Similarly, using also the fact that $\phi_1(x) < \phi_2(x)$ for all $x\in \mR$, we can prove $0 < Q < 1/2$
and $Q_\perp > 0$.
Finally, we have  
$$
 Q + Q_\perp = \frac{\cI_0^1}{\cI_0^2} < 1
$$
(also following from $\phi_1(x) < \phi_2(x)$), which proves $Q_\perp < 1-Q$.
\end{proof}
\begin{remark}
\rm
All moments of the form $\bk{\nu_1^\ell \nu_2^m f_\eq}$ can be expressed in terms of the functions $\cI_N^s$.
In fact, similarly to what we have done to derive Eq.\ \eqref{PQ}, this expression can be reduced to a linear combination
of integrals of the form
$$
  \frac{1}{\pi}\int_0^\pi  r(\cos\theta) \phi_2(A+B\cos\theta) \,d\theta,
$$
where $r$ is a polynomial. 
Then, by using
$$
  r(\cos\theta) = \sum_{N=0}^{N(r)} \alpha_N(r)\, T_N(\cos\theta) 
  = \sum_{N=0}^{N(r)} \alpha_N(r)\cos(N\theta),
$$
where $T_N$ are the Chebyschev polynomial of the first kind (and the coefficients $\alpha_N(r)$ can
be computed, e.g., by means of the Clenshaw algorithm \cite{Fox68}), we see tat the above integral can be written as 
$ \sum_{N=0}^{N(r)} \alpha_N(r) \cI_N^2(A,B)$.
\end{remark}
\subsection{Series expansion of $\cI_N^s(A,B)$}
\label{Sec_series}
In this section we shall make use of the extension of the Fermi functions $\phi_s(x)$ to negative
values of $s$. 
In order to understand this extension, we recall that, for $s>0$, the Fermi integral 
\eqref{phidef} can be expressed as
\BE
\label{phiLi}
  \phi_s(x) = -\Li_s(-\e^x), \qquad x \in \mR,
\EE
where $\Li_s(z)$ denotes the polylogarithm of order $s$ \cite{Lewin81}.
The latter is defined, for all  $s \in \mR$ and $z$ in the complex unit disc,
by the power series
\BE
\label{pwrLi}
  \Li_s(z) = \sum_{k=1}^\infty \frac{z^k}{k^s}, \qquad \abs{z} < 1, 
\EE
and can be analytically continued to a larger domain (depending on $s$) which includes the real semi-axis 
$z \in (-\infty,1)$. 
Then, Eq.\ \eqref{phiLi} provides a definition of $\phi_s(x)$ as an analytic function of $x\in \mR$ 
for every $s\in \mR$. 
The power series expansion of $\phi_s(x)$ at $x=0$ converges for $\abs{x} < \pi$ and reads 
as follows \cite{Lewin81,Wood92}:
\BE
\label{psephi}
  \phi_s(x) = \sum_{k=0}^\infty  \frac{h(s-k)}{k!}\,x^k,
  \qquad \abs{x}<\pi,
\EE
where
$$
    h(s) = -\Li_s(-1) = (1-2^{1-s})\zeta(s)
$$
($\zeta$ denoting the Riemann zeta function).
\par
A series expansion of $\cI_N^s(A,B)$ can be easily obtained by considering the derivatives of
\eqref{Idef} with respect to $B$ at $B=0$:
$$
  \frac{\pt^j \cI_N^s}{\pt B^j}(A,0) = 
  \frac{\phi_{s-j}(A)}{\pi}  \int_0^{\pi} \cos(N\theta) \cos^j\!\theta\, d\theta
$$
(where property \eqref{deriphi}, which extends to every $s\in\mR$, was used).
But 
\begin{multline}
\label{BesselCoeff}
   \frac{1}{\pi}  \int_0^{\pi} \cos(N\theta) \cos^j\!\theta\, d\theta
   = \frac{\pt^j I_N}{\pt x^j}(0)
\\[4pt]
   = \left\{\begin{aligned}
&0,& &\text{if $0\leq j \leq N-1$,} 
\\
&0,& &\text{if $j = N+2n+1$, $n\geq 0$,} 
\\
&\frac{2^{-(2n+N)}}{(N+n)!},&\ &\text{if  $j = N+2n$, $n\geq 0$,} 
\end{aligned}
\right.
\end{multline}
as it follows from well known properties of the modified Bessel function of the first kind $I_N(x)$.
We obtain therefore the series
\BE
\label{Iexp1}
   \cI_N^s(A,B) = 
   \sum_{n = 0}^\infty  \frac{\phi_{s-2n-N}(A)}{n!\,(N+n)!} \left(\frac{B}{2}\right)^{N+2n},
\EE
whose convergence is uniform on every compact set in the $(A,B)$ plane.
\par
If, moreover, the expansion \eqref{psephi} is used, we can write the following power-series expansion
\BE
\label{Iexp2}
   \cI_N^s(A,B) = 
   \sum_{n,k = 0}^\infty  \frac{h(s-N-2n-k)}{k!\,n!\,(N+n)!\, 2^{N+2n}}\, A^k B^{N+2n},
   \qquad
   \abs{A} < \pi,
\EE
whose convergence is uniform in the compact sets of $(-\pi,\pi)\times \mR$.
\section{Asymptotic regimes}
\label{Sec_asymp}
Although in general we are not able to give the fluid equations \eqref{S1} an explicit form
(by which we basically mean that $P_{ij}$ and $Q_{ij}$ are expressed as functions of $n$ and $\uu$
in terms of elementary functions), it is nevertheless possible to find explicit asymptotic forms of 
equations \eqref{S1} in some particular regime, which will be considered in this section.
\subsection{Diffusive limit ($\abs{\uu}\to 0$)}
\label{Sec_diffu}
\subsubsection{Diffusive equations}
The diffusive regime corresponds to vanishing mean velocity $c\uu$, i.e.\ to the limit $\abs{\uu}\to 0$.
It is evident from \eqref{S1} that, without further assumptions, such limit  would lead to trivial fluid equations 
describing a fluid which does not evolve at all.
The reason is well known: Eq.\ \eqref{S1} has been obtained from Eq.\ \eqref{WE3} as a leading order 
approximation in the the hydrodynamic limit $\tau\to 0$ (see Sect.\ \ref{Sec_MomentClosure}) 
but, in order to observe the diffusion current, we have to look at the first-order in $\tau$ (this is the so-called
Chapman-Enskog expansion \cite{Cercignani88})
When doing so, however, we bring into play terms coming from $\ff_\perp$, because only at leading order 
(i.e.\ when $\ff = \ff_\eq$) such terms disappear. 
Unfortunately, the diffusive equations obtained in this way are singular and require, e.g., 
a parabolic regularization of the Hamiltonian \cite{Zamponi12}. 
This issue is certainly interesting, and worth a deeper investigation, but it goes beyond the scope of 
the present paper.
Then, let us follow here an alternative approach \cite{Lundstrom00}
by introducing in the moment equations \eqref{S1} 
a current-relaxation term that acts on a time-scale $\tau_0 \gg \tau$:
\BE
\label{MEscal}
\left\{
\begin{aligned}
&\pt_t n  + c\pt_i (n u_i) = 0,
\\[3pt]
&\pt_t (n u_i) + c\pt_j P_{ij}  = \pm F_jQ_{ij} - \frac{nu_i}{\tau_0}.
\end{aligned}
\right.
\EE
Rescaling time and direction field as 
$$
  t^* = \tau_0 t, \qquad \uu^* =  \uu/\tau_0,
$$
we obtain
\BE
\label{MEscal2}
\left\{
\begin{aligned}
&\pt_{t^*} n  + c\pt_i (n u_i^*) = 0,
\\[3pt]
&\tau_0^2 \,\pt_{t^*} (n u_i^*) + c\pt_j P_{ij}  = \pm F_jQ_{ij} - nu_i^*.
\end{aligned}
\right.
\EE
Note that $P_{ij}$ and $Q_{ij}$ (see definition \eqref{PQ}) remain unchanged under the scaling of $\uu$, 
except that the Lagrange multipliers have to satisfy
\BE
\label{ABdiff}
  \cI_0^2(A,B) = \frac{n}{n_T}, \qquad  \frac{\cI_1^2(A,B)}{ \cI_0^2(A,B)} = \tau_0\abs{\uu^*}.
\EE
As $\tau_0 \to 0$ we obtain the condition $\cI_1^2(A,B) = 0$, which is satisfied if and only if $B = 0$.
Then, by Eq.\ \eqref{Bat0}, $A$ is given by $\phi_2(A) = n/n_T$ and, in conclusion, we obtain
\BE
\label{invA}
  A = \phi_2^{-1}\left( \frac{n}{n_T} \right), \qquad B = 0.
\EE
Now, from \eqref{PQcoeff}, \eqref{Bat0} and \eqref{invA}, we have that
$P(A,0) = P_\perp(A,0) = \frac{1}{2}$ and 
$$
  Q(A,0) = Q_\perp(A,0) = \frac{\phi_1(A)}{2\phi_2(A)} 
  = \frac{n_T\,\phi_1\Big( \phi_2^{-1}\Big( \frac{n}{n_T} \Big) \Big)}{2n},
$$
which yields
$$
  P_{ij}(A,0) = \frac{n}{2}\, \delta_{ij},
  \qquad
  Q_{ij}(A,0) = \frac{c\, n_T}{2k_BT}\, \phi_1\Big( \phi_2^{-1}\Big( \frac{n}{n_T} \Big) \Big)\, \delta_{ij}.
$$
Then, letting $\tau_0 \to 0$ in Eq.\ \eqref{MEscal2}, we obtain the diffusive equation
\begin{multline*}
   \pt_{t^*} n = c \pt_i\left[ c\pt_j P_{ij}(A,0) \pm \pt_jV\, Q_{ij}(A,0)\right]
 \\[4pt]
   =   \frac{c^2}{2} \pt_i \left[\pt_i n 
   \pm \frac{n_T}{k_BT}\,\phi_1\Big( \phi_2^{-1}\Big( \frac{n}{n_T} \Big) \Big)\pt_i V \right],
\end{multline*}
(where we recall that $F_i = -\pt_iV$) that is, in terms of the original time variable,
\begin{multline}
\label{DEgeneral}
 \pt_t n =  \tau_0 c \pt_i\left[ c\pt_j P_{ij}(A,0) \pm \pt_jV\, Q_{ij}(A,0)\right]
 \\[4pt]
   = \frac{\tau_0 c^2}{2} \nabla\cdot \left[\nabla n 
   \pm \frac{n_T}{k_BT}\,\phi_1\Big( \phi_2^{-1}\Big( \frac{n}{n_T} \Big) \Big)\nabla V \right].
\end{multline}
This drift-diffusion equation has a nonconventional, and somehow specular, structure with respect to the
drift-diffusion equations for Fermions with parabolic dispersion relation 
\cite{JSP12,JungelKrausePietra11,TrovatoReggiani10}.
Indeed, the diffusion coefficient (which is proportional to the variance of the velocity distribution), is here
independent of the temperature $T$, because the particles move with constant speed $c$, while it is proportional 
to $T$ in the parabolic case.
On the contrary, the mobility coefficient (which is related to the distribution of the second derivative 
of the energy, i.e.\ to the effective-mass tensor) is here temperature-dependent 
while in the parabolic case is constant.
\subsubsection{Quasi-diffusive regime (linear response)}
The quasi-diffusive regime consists in a linear-response approximation with respect to $B$, 
which amounts to considering only first-order terms in $B$ in the expansion \eqref{Iexp1}.
We then have, up to $\cO(B^2)$,
\BE
\label{linres}
 \cI_0^s =   \phi_s(A), \qquad  \cI_1^s =   \frac{1}{2}\,\phi_{s-1}(A) B,
 \qquad \cI_N^s = 0, \quad N\geq 2.
\EE
Using this approximation in Eqs.\ \eqref{ABeq} and \eqref{PQ} yields
$$
  A = \phi_2^{-1}\left( \frac{n}{n_T} \right), \qquad 
  B = \frac{2n\abs{\uu}}{n_T\,\phi_1\Big( \phi_2^{-1}\Big( \frac{n}{n_T} \Big) \Big)},
$$
and
$$
  P_{ij} = \frac{n}{2}\, \delta_{ij},
  \qquad
  Q_{ij} = \frac{k_BT}{2\pi c\hbar^2}\, \phi_1\Big( \phi_2^{-1}\Big( \frac{n}{n_T} \Big) \Big)\, \delta_{ij}.
$$
By substituting these expressions in \eqref{S1}, and taking the derivative with respect to time of the continuity equation, 
we obtain a wave equation for $n$:
\BE
\label{Wave}
  \frac{\pt^2 n}{\pt t^2} = \frac{c^2}{2} \Delta n \pm \frac{k_BT}{2\pi \hbar^2}\,\nabla \cdot
  \left[ \phi_1\Big( \phi_2^{-1}\Big( \frac{n}{n_T} \Big) \Big)\nabla V  \right].
\EE 
\subsection{Maxwell-Boltzmann regime ($T\to +\infty$)}
\label{Sec_MB}
\subsubsection{Hydrodynamic equations}
Let us now consider the asymptotic form of system \eqref{S1}--\eqref{S4} for high temperature.
Since $n_T \sim T^2$ (recall definition \eqref{nT}), 
form the first of the constraint equations \eqref{ABeq} we obtain that
$$
   \cI_0^2(A,B) = \frac{n}{n_T} \to 0, \quad \text{as $T\to +\infty$}.
$$
According to the discussion performed in the proof of Theorem \ref{solvability}, this 
implies that $A^2+B^2 \to \infty$ with $A<-B$ and then, according to Lemma \ref{Asym}, 
we can use the asymptotic approximation 
\BE
\label{MBapprx}
   \cI_N^s(A,B) \sim \e^A I_N(B),
\EE
where $I_N$ denotes the modified Bessel function of the first kind and where, remarkably, 
the dependence on $s$ disappears.
It can be easily seen that this corresponds to approximating $f_\eq$ with the Maxwell-Boltzmann 
distribution
\BE
\label{MB}
  f_\eq \sim   \e^{-\frac{c}{k_BT} \abs{\pp} + \nnu\cdot\BB + A}.
\EE
By using \eqref{MBapprx}, the constraint equations \eqref{ABeq} become
\BE
\label{ABMB}
  \e^A I_0(B) = \frac{n}{n_T}, \qquad \frac{I_1(B)}{I_0(B)} = \abs{\uu}.
\EE
Note, in particular, that $B$ only depends on $\abs{\uu}$ (this can be well visualized in Figure \ref{figure1}
in the region below the critical line $A = -B$ where, for large $A^2+B^2$, the level lines of $u = \abs{\uu}$ are parallel 
to the $A$-axis).
Moreover, the coefficients $P$, $P_\perp$, $Q$ and $Q_\perp$ take the simple form
\BE
\label{PQMB_0}
   P = 1-P = Q_\perp = 1 - Q = \frac{I_0(B) + I_2(B)}{2I_0(B)}
\EE
(in particular, they only depend  on $\abs{\uu}$).
Then, by introducing the function $X(\abs{\uu})$ defined by
\BE
\label{Xdef}
  X(\abs{\uu}) = \frac{I_0(B) + I_2(B)}{2I_0(B)}, \qquad
   B = \Big(\frac{I_1}{I_0}\Big)^{-1}(\abs{\uu})
\EE
(where we recall that $0\leq \abs{\uu} < 1$), we obtain the Maxwell-Boltzmann asymptotic form 
of Eq.\  \eqref{PQ}:
\BE
\label{PQMB}
\begin{aligned}
&P_{ij} = \frac{n}{\abs{\uu}^2} \left[ X(\abs{\uu})u_iu_j + \left(1-X(\abs{\uu})\right)u_i^\perp u_j^\perp \right],
\\[4pt]
&Q_{ij} =  \frac{c}{k_BT}\frac{n}{\abs{\uu}^2} \left[ \left(1-X(\abs{\uu})\right) u_iu_j + X(\abs{\uu}) u_i^\perp u_j^\perp \right].  
\end{aligned}
\EE
The function $X(\abs{\uu})$ is plotted in Fig.\ \ref{figure2}. 
Note that  $X(\abs{\uu})$ increases monotonically from $1/2$ to $1$ as $\abs{\uu}$ increases from 0 to 1.
The two extreme points  $\abs{\uu} = 0$ and $\abs{\uu} = 1$ correspond, respectively, to $B = 0$ and $B \to +\infty$
and yield, respectively, the diffusive limit and the collimation limit, as discussed below.
\subsubsection{Diffusive equations}
This limit can be easily obtained from  the general drift-diffusion equation \eqref{DEgeneral},
either from the Maxwell-Boltzmann approximation of the functions $\phi_s$ ($\phi_s(x) \sim \e^x$)
of by writing $P_{ij}$ and $Q_{ij}$ with the coefficients $P$, $P_\perp$, $Q$ and $Q_\perp$ 
as given by \eqref{PQMB_0}.
We obtain in this way the following drift-diffusion equation:
\BE
\label{DEMB}
 \pt_t n = \frac{\tau_0 c^2}{2} \nabla\cdot 
 \left(\nabla n \pm \frac{n}{k_BT}\,\nabla V \right).
\EE
In addition to the comments made about Eq.\ \eqref{DEgeneral}, we remark here the linear dependence 
of the mobility coefficient on $n$.
This reflects the linear dependence  on the density of the Maxwell-Boltzmann distribution. 
\subsubsection{Collimation limit ($\abs{\uu} \to 1$)}
\label{ColliMB}
On the opposite side with respect to the diffusive limit we find the limit of completely non-spread directions, 
corresponding to $\abs{\uu} \to 1$.
In other words, this is the limit in which the Maxwell-Boltzmann distribution \eqref{MB} becomes concentrated along a
($(\xx,t)$-dependent) direction in the $\pp$-space (the direction determined by $\uu$).
We term this regime the ``collimation limit''.
It is worth remarking that the collimation regime is only possible in the $T\to +\infty$ limit, considered here, and in the 
$T\to 0$ limit, to be considered next.
In fact, as it emerges from the proof of Theorem \ref{solvability}, $\abs{\uu} \to 1$ only when $A^2+B^2 \to \infty$ 
and the critical line $A = -B$ is approached from below or from above.
\par
Since $X(\abs{\uu}) \to 1$ as $\abs{\uu} \to 1$, from Eq.\ \eqref{PQMB} we obtain
$$
  P_{ij} \to nu_iu_j, \qquad Q_{ij} \to \frac{c}{k_BT} n u_i^\perp u_j^\perp.
$$
Then, the second of Euler equations \eqref{S1} reduces to
$$
  \pt_t (n u_i) + c\pt_j (nu_iu_j)  = \pm \frac{c}{k_BT}\, nu_i^\perp F_ju_j^\perp,
$$
which, by using the continuity equation $\pt_t n + c\pt_j(nu_i)$, can  be rewritten as
\BE
\label{MBPE}
   \pt_t u_i + cu_j\pt_j u_i = \pm \frac{c}{k_BT}\,u_i^\perp F_ju_j^\perp.
\EE 
We see, therefore, that the equation for $\uu$ decouples from the continuity equation.
A simple computation shows that Eq.\ \eqref{MBPE} is compatible with the assumption of collimation.
In fact, multiplying by $2u_i$ both sides of \eqref{MBPE} and summing up over $i$ yields the equation
$$
  \pt_t \abs{\uu}^2 + c \uu \cdot \nabla \abs{\uu}^2 = 0,
$$
from which we see (e.g.\ using characteristics) that, for any regular solution $\uu(\xx,t)$ of Eq.\ \eqref{MBPE} such that 
$\abs{\uu(\xx,0)} = 1$ for all $\xx$, we have $\abs{\uu(\xx,t)} = 1$ for all $(\xx,t)$ where the solution exists. 
\begin{remark}
\label{GeomRem}
\rm
Let us consider the stationary version of Eq.\ \eqref{MBPE}, that we rewrite as follows:
$$
   (\uu\cdot \nabla) \uu \pm (\uu^\perp\cdot\nabla K)\uu^\perp = 0,
   \qquad 
   K = \frac{V}{k_BT}.
$$ 
If $\abs{\uu} = 1$, by substituting  $\uu = (\sin\varphi,\cos\varphi)$ in the last equation we obtain
the conservation law
\BE
\label{geom_optics}
  \nabla\cdot\left(\e^{\mp K} \uu_\perp\right) = 0.
\EE
Equation \eqref{geom_optics} reveals that the collimation regime has the properties of a geometrical-optics system,
with $\e^{\mp K}$ playing the role of the refractive index. 
For example, if $V = V(x_1)$ is a potential step of height $\delta V$ at $x_1 = 0$, Eq.\ \eqref{geom_optics}
implies the Snell law for incident and refracted angles:
$$
     \frac{\sin\varphi_i}{\sin\varphi_r} = \e^{\mp\frac{\delta V}{k_BT}}
$$
(in particular, the region at higher potential has lower refractive index for electrons, and conversely for holes).
Note that, within our semiclassical model, we always find a positive refractive index; 
however, the possibility if a negative refractive index arises from a fully quantum description \cite{CheianovEtAl2007}.
\end{remark}
\subsection{Degenerate gas limit ($T\to 0$)}
\label{Sec_DG}
\subsubsection{Hydrodynamic equations}
We now consider the limit of the fluid model \eqref{S1}--\eqref{S4} when $T \to 0$, describing a so-called
degenerate electron/hole gas.
Since we now have
$$
   \cI_0^2(A,B) = \frac{n}{n_T} \to +\infty, \quad \text{as $T\to 0$},
$$
then we know from Sect.\ \ref{Sec_solvability} that $A^2+B^2 \to \infty$ with $A>-B$, and we can use the 
asymptotic approximation
\BE
\label{T0apprx}
\cI_N^s(A,B) \sim
 \frac{1}{\pi\Gamma(s+1)} \int_0^{C(A,B)} \hspace{-16pt} \cos(N\theta) (A+B\cos\theta)^s d\theta,
\EE
where $C(A,B)$ is given by \eqref{Cdef}.
It is convenient to introduce polar coordinates in the $(A,B)$-plane,
$$
   A = R\cos\psi, \qquad B = R\sin\psi,
$$
and rewrite \eqref{T0apprx} as follows:
\BE
\label{T0apprxPolar}
\cI_N^s(R\cos\psi,R\sin\psi) \sim R^s \cF_N^s(\psi), \qquad R > 0, \quad 0 \leq  \psi  < \frac{3\pi}{4},
\EE
where
\BE
\label{cFdef}
\cF_N^s(\psi) =  \frac{1}{\pi\Gamma(s+1)} \int_0^{C(\psi)} \hspace{-16pt} \cos(N\theta) (\cos\psi+\sin\psi\,\cos\theta)^s d\theta
\EE
and
\BE
\label{Cdef2}
 C(\psi)
   = \left\{
  \begin{aligned}
  &\arccos\left(-\cot\psi \right),& &\text{if $\frac{\pi}{4} < \psi < \frac{3\pi}{4}$},
  \\[2pt]
  &\pi,& &\text{if $0\leq \psi \leq \frac{\pi}{4}$}.
  \end{aligned}
  \right.
\EE
The asymptotic form of the constraint equations \eqref{ABeq} is now 
\BE
\label{ABT0}
  R^2\cF_0^2(\psi) = \frac{n}{n_T}, \qquad \frac{\cF_1^2(\psi)}{\cF_0^2(\psi)} = \abs{\uu},
\EE
from which we see, in particular, that $\abs{\uu}$ only depends on $\psi$.
Moreover, \eqref{PQcoeff} becomes
\BE
\label{PQcoeff2}
  P = \frac{\cF_0^2+\cF_2^2}{2\cF_0^2} = 1 - P_\perp,
  \quad
  Q = \frac{1}{R} \frac{\cF_0^1-\cF_2^1}{2\cF_0^2},
  \quad
  Q_\perp =  \frac{1}{R} \frac{\cF_0^1+\cF_2^1}{2\cF_0^2}.
\EE
Using \eqref{ABT0} and \eqref{PQcoeff2}, we obtain the asymptotic form of the terms $P_{ij}$ and $Q_{ij}$
for the degenerate gas:
\BE
\label{PQT0}
\begin{aligned}
&P_{ij} = \frac{n}{\abs{\uu}^2} \left[ Y(\abs{\uu})u_iu_j + \left(1-Y(\abs{\uu})\right)u_i^\perp u_j^\perp \right],
\\[4pt]
&Q_{ij} =  \frac{\sqrt{n}}{\hbar\sqrt{\pi}\abs{\uu}^2} 
\left[ Z(\abs{\uu}) u_iu_j + Z_\perp(\abs{\uu}) u_i^\perp u_j^\perp \right], 
\end{aligned}
\EE
where the functions $Y( \abs{\uu})$, $Z(\abs{\uu})$ and  $Z_\perp(\abs{\uu})$ are defined, for  
$0\leq \abs{\uu} < 1$,
by
\BE
\label{YZdef}
\begin{gathered}
  Y(\abs{\uu}) = \frac{\cF_0^2(\psi) + \cF_2^2(\psi) }{2\cF_0^2(\psi)},
\\[4pt]
  Z(\abs{\uu}) = \frac{\cF_0^1(\psi) - \cF_2^1(\psi) }{2\sqrt{2\cF_0^2(\psi)}},
\qquad
Z_\perp(\abs{\uu}) = \frac{\cF_0^1(\psi) + \cF_2^1(\psi) }{2\sqrt{2\cF_0^2(\psi)}},
\end{gathered}
\EE
and
$$
   \psi = \Big(  \frac{\cF_1^2}{\cF_0^2}  \Big)^{-1}(\abs{\uu}).
$$
A plot of the functions  $Y$, $Z$ and $Z_\perp$ is shown in Figure \ref{figure2}. 
\begin{figure}[h]
\includegraphics[width=\linewidth]{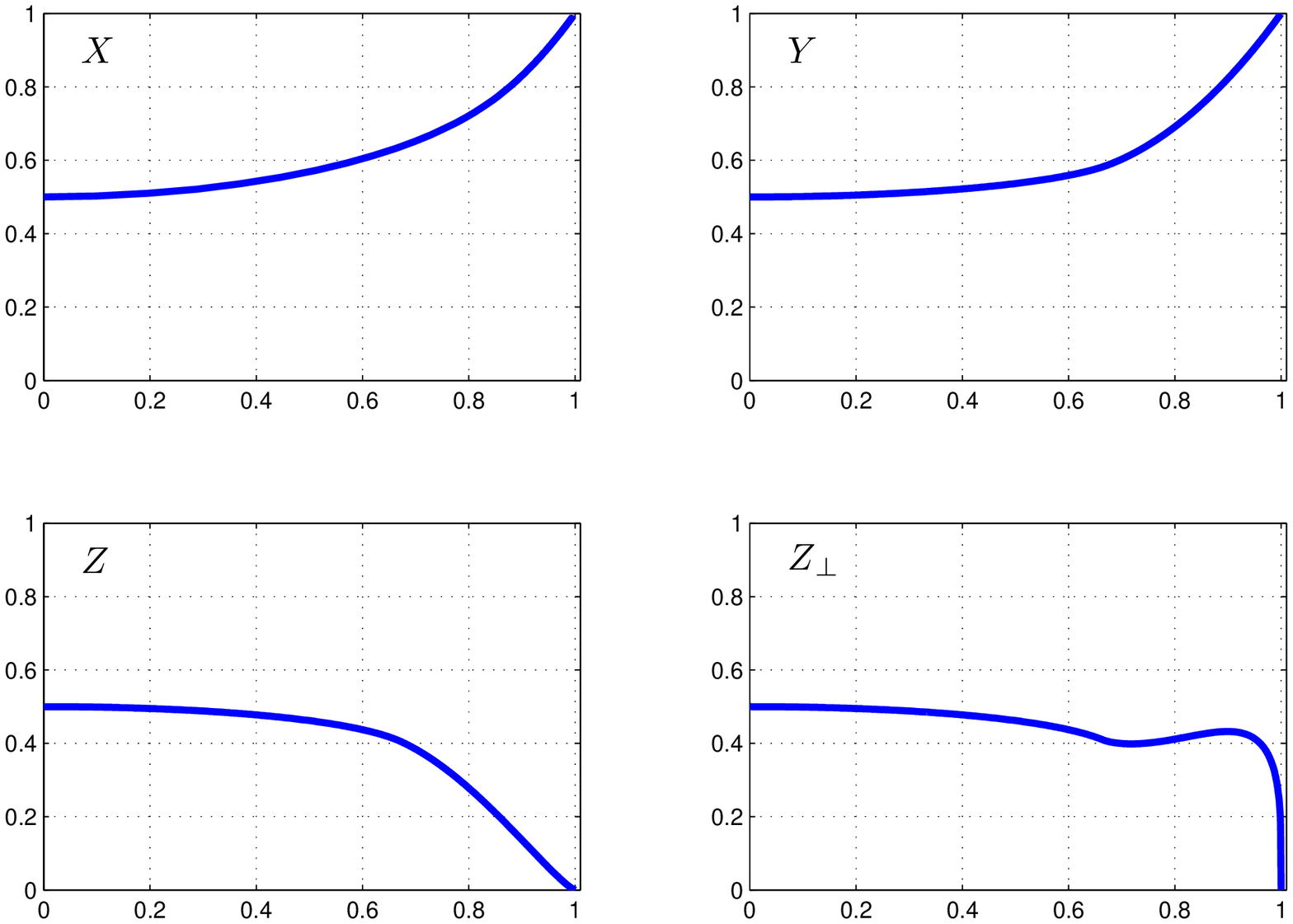}
\caption{Plot of the functions $X(\abs{\uu})$, $Y(\abs{\uu})$, $Z(\abs{\uu})$ and $Z_\perp(\abs{\uu})$
(see definitions \eqref{Xdef} and  \eqref{YZdef}).}
\label{figure2}
\end{figure}
Moreover, the behaviour of $Y$, $Z$ and $Z_\perp$ in the two limits $\abs{\uu} \to 0$ (diffusive) and
$\abs{\uu} \to 1$ (collimation) is discussed in the following Theorem.
\begin{theorem}
\label{Theo_YZ}
For $\abs{\uu} \to 0$ we have
\BE
\label{YZu0}
\begin{aligned}
   Y(\abs{\uu})& = \frac{1}{2} + \frac{1}{8}\abs{\uu}^2 + \cO(\abs{\uu}^4),
\\[4pt]
  Z(\abs{\uu})& =  \frac{1}{2} - \frac{1}{8}\abs{\uu}^2
    + \cO(\abs{\uu}^4).
\\[4pt]  
    Z_\perp(\abs{\uu})& = \frac{1}{2} - \frac{1}{8}\abs{\uu}^2
    + \cO(\abs{\uu}^4).
  \end{aligned}
\EE
For $\abs{\uu} \to 1$ we have
\BE
\label{YZu1}
 \begin{aligned}
   Y(\abs{\uu})& = 2\abs{\uu} - 1 + \cO\big((1-\abs{\uu})^2\big),
   \\[4pt]
  Z(\abs{\uu})&  = \frac{(14)^\frac{5}{4}}{\sqrt{30\pi}}\,(1 - \abs{\uu})^\frac{5}{4}
  + \cO\big((1-\abs{\uu})^\frac{9}{4}\big),
    \\[4pt]
    Z_\perp(\abs{\uu})& = \frac{\sqrt{5}\, (14)^\frac{1}{4}}{\sqrt{6\pi}}\,(1 - \abs{\uu})^\frac{1}{4}
  + \cO\big((1-\abs{\uu})^\frac{5}{4}\big).
  \end{aligned}
\EE
\end{theorem}
\begin{proof}
For $\abs{\uu} \to 0$, we already know that $\psi \to 0$.
Then, we have $C(\psi) = \pi$ and Eq.\ \eqref{cFdef} can be rewritten as follows:
$$
  \cF_N^s(\psi) =  \frac{\cos^s\!\psi}{\Gamma(s+1)} \sum_{j=0}^s \binom{s}{j} \tan^j\!\psi 
  \,\frac{1}{\pi}\! \int_0^\pi\cos(N\theta) \cos^j\!\theta \,d\theta,
$$
where the integral is a Bessel coefficient, given by Eq.\ \eqref{BesselCoeff}.
This allows to easily compute the Taylor expansion of $\cF_N^s(\psi)$ and of the associated functions. 
We obtain, in particular:
$$
\begin{gathered}
    \frac{\cF_1^2(\psi)}{\cF_0^2(\psi)}  = \psi + \cO(\psi^3),
  \qquad
    \frac{\cF_0^2(\psi) + \cF_2^2(\psi) }{2\cF_0^2(\psi)} = \frac{1}{2} + \frac{1}{8}\psi^2 + \cO(\psi^4),
\\[4pt]
     \frac{\cF_0^1(\psi) \pm \cF_2^1(\psi) }{2\sqrt{2\cF_0^2(\psi)}} = \frac{1}{2} - \frac{1}{8}\psi^2 + \cO(\psi^4),
\end{gathered}
$$
from which, using \eqref{ABT0} and \eqref{YZdef}, Eq.\ \eqref{YZu0} follows.
\par
For $\abs{\uu} \to 1$, we know that $\psi \to \frac{3\pi}{4}$. 
Then $C(\psi) = \arccos\left(-\cot\psi\right)$, and $C$ can be used as independent variable
 (note that the limit $\psi \to \frac{3\pi}{4}^-$ corresponds to $C \to 0^+$).
 Equation \eqref{cFdef} is therefore rewritten as 
 $$
 \cF_N^s =  \frac{\sin^s\!\psi}{\pi\Gamma(s+1)} \int_0^{C} \hspace{0pt} \cos(N\theta) (\cos\theta - \cos C)^s d\theta,
$$
which makes easier the computation of Taylor expansions (around $C = 0$) in this case.
In particular, after some straightforward algebra, we obtain
$$
\begin{aligned}
    &\frac{\cF_1^2}{\cF_0^2}  = 1 - \frac{1}{14}C^2+ \cO(C^4),
   &\quad
   &\frac{\cF_0^2+ \cF_2^2 }{2\cF_0^2} = 1 - \frac{1}{7}C^2 + \cO(C^4),
\\[4pt]
   &\frac{\cF_0^1 + \cF_2^1 }{2\sqrt{2\cF_0^2}} = \frac{\sqrt{5}}{\sqrt{6\pi}} C^\frac{1}{2} + \cO(C^\frac{5}{2}),
  &\quad
   &\frac{\cF_0^1 - \cF_2^1 }{2\sqrt{2\cF_0^2}} = \frac{1}{\sqrt{30\pi}} C^\frac{5}{2} + \cO(C^\frac{9}{2}),
\end{aligned}
$$
which, using \eqref{ABT0} and \eqref{YZdef}, yields Eq.\ \eqref{YZu1}.
\end{proof}
\subsubsection{Diffusive equations}
As already remarked in the Maxwell-Boltzmann case, also the diffusive limit for the degenerate gas 
can be obtained from  the general drift-diffusion equation \eqref{DEgeneral},
either by the vanishing-temperature approximation of the functions $\phi_s$ 
($\phi_s(x) \sim \frac{x^s}{\Gamma(s+1)}$)
of by writing $P_{ij}$ and $Q_{ij}$ with the coefficients $P$, $P_\perp$, $Q$, $Q_\perp$ 
as given by \eqref{PQcoeff2}.
The drift-diffusion equation we get is the following:
\BE
\label{DEDG}
 \pt_t n = \frac{\tau_0 c}{2}\, \nabla\cdot 
 \left(c\nabla n \pm \frac{1}{\hbar\,\sqrt{\pi}}\,\sqrt{n}\,\nabla V \right).
\EE
This has to be compared with the diffusive equations for a degenerate Fermi gas 
with standard (parabolic) dispersion relation
\cite{JSP12,JungelKrausePietra11,TrovatoReggiani10} and, again, we remark that the conical dispersion
relation ``inverts'' the structure of diffusion and mobility coefficients. 
In particular, here, the nonlinear dependence on $n$ is in the mobility term and not in the diffusion term, 
as it happens in the parabolic case.
\subsubsection{Collimation limit}
Contrarily to what happens in the Maxwell-Boltzmann case (Sect.\ \eqref{ColliMB}), the collimation
limit for the degenerate gas leads to trivial equations, to the extent that all the force terms vanish.
This can be easily seen from Eq.\ \eqref{YZu1}, which implies that
$Y(\abs{\uu}) \to 1$, $Z(\abs{\uu}) \to 0$ and  $Z_\perp(\abs{\uu}) \to 0$, as $\abs{\uu} \to 1$. 
Then we obtain $P_{ij} \to nu_iu_j$ and  $Q_{ij} \to 0$, and system \eqref{S1} reduces to the inviscid 
Burger's equation $\pt_t u_i + cu_j\pt_j u_i = 0$, independently of the force field.
\subsection*{Acknowledgements}
This work was supported by MIUR National Project {\em Kinetic and hydrodynamic equations of complex collisional 
system} (PRIN 2009, Prot.\ n.\ 2009NAPTJF\_003) as well as by INdAM-GNFM, Progetto Giovani Ricercatori 2013 
{\em Quantum fluid-dynamics of identical particles:  analytical and numerical study}.

\end{document}